\documentclass[conference]{IEEEtran}

\IEEEoverridecommandlockouts
\usepackage[letterpaper, left=0.65in, right=0.65in, bottom=1.02in, top=0.75in]{geometry}
\usepackage{amsmath,amssymb,amsfonts}
\usepackage{amsthm}
\usepackage{cite}
\usepackage{graphicx}
\usepackage{xcolor}
\usepackage{tikz}
\usetikzlibrary{calc,decorations.pathreplacing,arrows.meta,positioning}

\newtheorem{theorem}{Theorem}
\newtheorem{corollary}{Corollary}
\theoremstyle{remark}

\begin{document}

\title{Switched-Feed Pinching-Antenna Systems for Wideband Terahertz Communications\\
\thanks{This work was supported by the German Federal Ministry of Education and Research (BMBF) through \emph{Open6GHub+} (Grant no.  \emph{16KIS2402K}) project.}
}

\author{
\IEEEauthorblockN{Wei Jiang}
\IEEEauthorblockA{Intelligent Networking Research Group\\
German Research Center for Artificial Intelligence (DFKI)\\
Kaiserslautern, 67663 Germany}
\and
\IEEEauthorblockN{Hans D. Schotten}
\IEEEauthorblockA{Department of Electrical and Computer Engineering\\
University of Kaiserslautern (RPTU)\\
Kaiserslautern, 67663 Germany}
}

\maketitle
\vspace{-2.5em}
\setlength{\textfloatsep}{6pt plus 2pt minus 2pt}
\setlength{\intextsep}{6pt plus 2pt minus 2pt}
\setlength{\abovecaptionskip}{4pt}
\setlength{\belowcaptionskip}{2pt}

\begin{abstract}
The pinching-antenna system (PASS) uses dielectric particles along a low-loss waveguide as reconfigurable passive radiators. Existing analyses conclude that the in-waveguide attenuation is negligible at low frequencies and millimeter wave bands; we show this fails at terahertz (THz), where realizable waveguide losses are dramatically larger. We develop a unified wideband THz-PASS propagation model integrating in-waveguide attenuation, atmospheric absorption, molecular re-radiation noise, and beam squint. Closed-form results follow: a band-averaged coherence factor; a cluster-center placement satisfying a band-edge SINR equalization condition; an associated placement-inversion threshold; and a proposed \emph{Switched-Feed PASS} (SF-PASS) architecture in which a centrally located radio-frequency switch routes the signal among multiple waveguide segments, with a closed-form insertion-loss payoff threshold. Numerical evaluation at the best PASS-compatible THz operating point shows that SF-PASS substantially outperforms single-feed PASS in spectral efficiency and is competitive with a large-scale antenna array at much lower hardware costs.
\end{abstract}


\section{Introduction}\label{sec:introduction}
\IEEEPARstart{T}{erahertz} (THz) frequencies (generally covering 0.1--3\,THz) are
regarded as a key enabler of next-generation wireless systems, offering
unprecedented bandwidths for terabit-per-second communication links~\cite{Ref_jiang2024thz6g},
at the cost of severe free-space spreading loss, line-of-sight blockage, and
frequency-resonant atmospheric absorption~\cite{Ref_jornet2011channel}. Conventional
solutions deploy large-scale arrays of hundreds of active antennas to produce high-gain beams,
with their attendant per-element radio frequency (RF)-chain or phase-shifter hardware cost.

The \emph{pinching-antenna system} (PASS), prototyped by NTT
DOCOMO~\cite{Ref_fukuda2022pinching} and developed as a wireless
architecture in~\cite{Ref_ding2025flexible,Ref_liu2026passtutorial,Ref_wang2025modeling},
offers a fundamentally different approach: a single RF chain feeds a long, low-loss
dielectric waveguide (WG) along which small dielectric ``pinches'' radiate, with pinch
positions reconfigurable at deployment to minimize the user-to-radiator distance.
In prior works, in-waveguide attenuation $\alpha_{\rm g}$ is often neglected
and a recent analysis~\cite{Ref_xu2025inwgblockage} concluded that in-waveguide
attenuation is negligible even under millimeter wave (mmWave) frequencies
(e.g., $\alpha_{\rm g}\!=\!0.0092$~Np/m, $\approx\!0.08$~dB/m at 28~GHz). However, realizable
THz dielectric-waveguide structures exhibit $\alpha_{\rm g}$ values that are
substantially larger---ranging from $0.43$~dB/m for porous-core polyethylene
fibers at 300~GHz to $87$~dB/m for microstructured high-density polyethylene at
1~THz~\cite{Ref_atakaramians2013thz}---invalidating the negligible-loss assumption at low and mmWave frequencies. In addition, three THz-unique effects also lack treatment in PASS literature:
frequency-resonant atmospheric absorption per
ITU-R~P.676~\cite{Ref_itur_p676_13}, molecular re-radiation
noise~\cite{Ref_jornet2011channel, Ref_tarboush2021teramimo}, and beam squint
across clustered pinches---qualitatively distinct from the array-aperture
squint of fixed arrays.

This paper develops a unified wideband THz-PASS framework with five
contributions: (i)~a unified wideband channel model integrating in-WG attenuation and the three THz-unique propagation effects mentioned above (Section~\ref{sec:channel_model});
(ii)~a closed-form, band-averaged coherence factor that quantifies the
beam-squint penalty of a clustered multi-pinch aperture
(Section~\ref{sec:beam_squint});
(iii)~a closed-form cluster-center placement rule based on band-edge SINR
equalization, together with a closed-form \emph{placement-inversion
threshold} above which the optimal cluster-center migrates from the user
projection toward the waveguide feed (Section~\ref{sec:optimal_placement});
(iv)~a structurally distinct \emph{Switched-Feed PASS} (SF-PASS) architecture
that employs a centrally located RF switch to route the
signal among multiple waveguide segments, together with a closed-form
insertion-loss payoff threshold that establishes when SF-PASS dominates
the single-feed counterpart (Section~\ref{sec:sfpass}); and
(v)~a numerical study in an indoor corridor at a representative
PASS-compatible THz operating point, showing that SF-PASS substantially
outperforms single-feed PASS in spectral efficiency (SE) and is competitive
with a uniform planar array (UPA) at \emph{much lower}
hardware costs (Section~\ref{sec:numerical}).

\section{System Model}\label{sec:system_model}

\begin{figure}[t]
\centering
\resizebox{\columnwidth}{!}{%
\begin{tikzpicture}[
  every node/.style={font=\footnotesize},
  axis/.style={->,>={Stealth[length=7pt,width=5pt]},very thick,black},
  pa/.style={circle,fill=blue!60!black,inner sep=0pt,minimum size=3pt},
  molecule/.style={circle,draw=orange!70!black,fill=orange!15,inner sep=0pt,minimum size=8pt,font=\tiny}
]
% Cabinet projection: y'-axis at 30 deg, foreshortening 0.5
\def\ya{cos(30)*0.5}
\def\yb{sin(30)*0.5}
\pgfmathsetmacro{\yax}{\ya}
\pgfmathsetmacro{\yay}{\yb}
\def\yspan{1.8}
\def\dh{2.5}
\def\xf{3.5}      % switch / feed position at L/2
\def\xc{4.7}      % PA cluster center (on active right segment)
\def\xu{5.5}
\def\yu{1.5}

% Coordinate axes (z extended up to give ITU inset more headroom)
\draw[axis] (0,0) -- (7.8,0) node[right,black] {$x$};
\draw[axis] (0,0) -- ({1.4*\yspan*\yax},{1.4*\yspan*\yay}) node[above right,black] {$y$};
\draw[axis] (0,0) -- (0,4.7) node[above,black] {$z$};
\node[below left,font=\scriptsize] at (0,0) {$O$};

% Ground plane (light gray parallelogram); label placed away from axes
\fill[black!5] (0,0) -- (7.2,0) --
  ({7.2+\yspan*\yax},{\yspan*\yay}) -- ({\yspan*\yax},{\yspan*\yay}) -- cycle;
\node[black!50,font=\scriptsize] at (3.2,0.20) {ground plane $z\!=\!0$};

% LEFT waveguide segment (idle, lighter solid blue — same shape as right, just lighter)
\draw[line width=2pt,blue!22!white!95!black] (0.2,\dh) -- ({\xf-0.18},\dh);
\node[blue!50!gray,font=\scriptsize] at (1.55,{\dh-0.25}) {idle segment $\mathcal{S}_1$};

% RIGHT waveguide segment (active, solid full blue); label below WG near right end, aligned with idle's S_1 label
\draw[line width=2.5pt,blue!55!black] ({\xf+0.18},\dh) -- (7.0,\dh);
\node[blue!55!black,font=\scriptsize] at (6.30,{\dh-0.25}) {active segment $\mathcal{S}_2$};

% RF switch at center (x_f = L/2) — small white box with SW label
\fill[white] ({\xf-0.18},{\dh-0.13}) rectangle ({\xf+0.18},{\dh+0.13});
\draw[blue!55!black,line width=0.8pt] ({\xf-0.18},{\dh-0.13}) rectangle ({\xf+0.18},{\dh+0.13});
\node[blue!55!black,font=\tiny] at (\xf,\dh) {SW};

% BS / RF chain box ABOVE switch, connected by vertical cable — lighter fill, larger font
\fill[blue!18!white] ({\xf-0.48},{\dh+0.82}) rectangle ({\xf+0.48},{\dh+1.32});
\draw[blue!55!black,line width=0.6pt] ({\xf-0.48},{\dh+0.82}) rectangle ({\xf+0.48},{\dh+1.32});
\node[blue!40!black,font=\scriptsize] at (\xf,{\dh+1.07}) {BS+RF};
\draw[blue!55!black,line width=1.2pt] (\xf,{\dh+0.82}) -- (\xf,{\dh+0.13});

% Switch position indicator on x-axis
\draw[dashed,black!50] (\xf,\dh-0.15) -- (\xf,0);
\node[below=2pt,font=\scriptsize] at (\xf,0) {$x_f\!=\!L/2$};

% Height indicator d
\draw[<->,thin,black!60] (-0.25,0.05) -- (-0.25,\dh-0.05) node[midway,left,font=\scriptsize] {$d$};

% In-waveguide attenuation on active segment (fading arrows from switch toward cluster)
\foreach \i [evaluate=\i as \op using {0.95-0.22*\i}] in {0,1,2,3} {
  \draw[->,green!50!black,opacity=\op,line width=1pt]
    ({3.80+0.13*\i},\dh+0.18) -- ({3.93+0.13*\i},\dh+0.18);
}
\node[green!45!black,above,font=\scriptsize] at (3.95,\dh+0.32) {$\alpha_{\rm g}(f)$};

% PA cluster at x_c
\foreach \i in {-2,-1,0,1,2} {
  \node[pa] at ({\xc+0.16*\i},\dh) {};
}
% Brace ABOVE cluster (to keep label area below clear for LoS + H2O)
\draw[decorate,decoration={brace,amplitude=3pt,raise=3pt}]
  ({\xc-0.36},\dh) -- ({\xc+0.36},\dh);
% N_act PAs label placed clearly RIGHT of cluster, well clear of alpha_g(f) and LoS
\node[font=\scriptsize,anchor=west,black!85] at ({\xc+0.55},{\dh+0.32}) {$N_{\rm act}$ PAs};
\draw[->,thin,black!55,>={Stealth[length=3pt]}]
  ({\xc+0.50},{\dh+0.30}) -- ({\xc+0.10},{\dh+0.18});

% x_c indicator (dashed drop to x-axis)
\draw[dashed,black!50] (\xc,\dh) -- (\xc,0);
\node[below=2pt,font=\scriptsize] at (\xc,0) {$x_c$};

% User position at (x_u, y_u, 0)
\coordinate (user) at ({\xu+\yu*\yax},{\yu*\yay});
\fill[red!75!black] (user) circle (3pt);
\node[red!75!black,below right=0pt and -3pt,font=\scriptsize] at (user) {UE $(x_u, y_u, 0)$};
% Drop lines x_u, y_u
\draw[dashed,black!50] (\xu,0) -- (user);
\draw[dashed,black!50] ({\yu*\yax},{\yu*\yay}) -- (user);
\node[below,font=\scriptsize] at (\xu,-0.02) {$x_u$};
\node[left,font=\scriptsize] at ({\yu*\yax-0.05},{\yu*\yay+0.05}) {$y_u$};

% LoS link from cluster to user. Compact "r_c" label placed below the line at fraction 0.30, clear of both the active-segment label (above the WG) and the H2O molecules (further down the line). The caption identifies the line as the LoS link.
\draw[red!75!black,line width=0.9pt] (\xc,\dh) -- (user);
\node[red!75!black,font=\scriptsize,sloped,below,yshift=-1pt]
  at ($(\xc,\dh)!0.10!(user)$) {LoS link $r_c$};

% Atmospheric molecules along LoS (concentrated toward UE end to leave the upper portion of the line free for the r_c label)
\foreach \pos in {0.55, 0.72, 0.88} {
  \node[molecule] at ($(\xc,\dh)!\pos!(user)$) {\tiny H$_2$O};
}

% Re-radiation noise arrows from molecules
\foreach \pos in {0.55, 0.72, 0.88} {
  \draw[->,orange!75!black,dashed,thin,>={Stealth[length=3pt]}]
    ($(\xc,\dh)!\pos!(user)+(0.18,-0.10)$) -- ($(\xc,\dh)!\pos!(user)+(0.46,-0.28)$);
}
\node[orange!75!black,font=\scriptsize,align=center] at (7.0,1.05)
  {mol.\ noise\\$N_{\rm mol}(f,r_c)$};

% Inset: ITU-R P.676 atmospheric attenuation curve
% Placed on the LEFT, on top of the idle segment, where the white space exists.
% Bands: 0.22 / 0.30 / 1.0 THz — VTC primary scenario; axis unit f [THz] (matches Table I).
\begin{scope}[shift={(0.20,3.40)},scale=0.70,every node/.style={font=\tiny}]
  % Inset bounding box (light frame)
  \draw[gray!40, line width=0.3pt] (-0.10,-0.30) rectangle (2.95,1.60);
  % ITU-R title (above curve, fully inside box)
  \node[above,font=\tiny,black!75] at (1.40,1.30) {ITU-R P.676};
  \draw[black!70,->,>={Stealth[length=2.5pt]}] (0,0) -- (2.75,0) node[right,xshift=-3pt,font=\tiny,black!85] {$f$\,[THz]};
  \draw[black!70,->,>={Stealth[length=2.5pt]}] (0,0) -- (0,1.30) node[right,xshift=2pt,yshift=-3pt,black!85] {$\alpha_{\rm atm}$};
  % Schematic ITU-R P.676-like curve (water-vapor + oxygen resonances above 100 GHz)
  \draw[blue!60!black,thick] plot[smooth,tension=0.6]
    coordinates {(0.0,0.05) (0.3,0.08) (0.5,0.13) (0.65,0.45) (0.78,0.15)
                 (0.95,0.18) (1.10,0.70) (1.22,0.25) (1.45,0.35) (1.60,0.85)
                 (1.78,0.45) (2.00,0.65) (2.18,0.95) (2.40,0.55)};
  % Mark the three operating windows (in THz: 0.22 / 0.30 / 1.0)
  \fill[red!75!black] (0.42,0) circle (1.6pt);
  \fill[red!75!black] (0.85,0) circle (1.6pt);
  \fill[red!75!black] (2.10,0) circle (1.6pt);
  % Stagger labels so 220 and 300 do not overlap
  \node[below=1pt,font=\tiny] at (0.32,-0.02) {0.22};
  \node[below=1pt,font=\tiny] at (0.95,-0.02) {0.30};
  \node[below=1pt,font=\tiny] at (2.10,-0.02) {1.0};
\end{scope}

\end{tikzpicture}%
}
\caption{System model of the proposed SF-PASS,
$K\!=\!2$ corridor. A centrally located RF switch at $x_f\!=\!L/2$ routes
the BS signal into the active waveguide segment $\mathcal{S}_2$ containing
the user (right, solid blue); the idle segment $\mathcal{S}_1$ (left,
lighter blue) is shown for context. %A cluster of $N_{\rm act}$ pinching antennas at position $x_c$ on the active segment serves a ground UE at $(x_u, y_u, 0)$. The four THz-unique physical effects integrated in the channel model are visualized: frequency-dependent in-WG attenuation $\alpha_{\rm g}(f)$ along the active segment (green), free-space LoS propagation with frequency-resonant atmospheric absorption and molecular re-radiation noise $N_{\rm mol}(f, r_c)$ from H$_2$O/O$_2$ along the link (orange), and beam squint across the cluster aperture. Inset: the ITU-R P.676 atmospheric attenuation spectrum, with the three operating windows used in Section~\ref{sec:numerical}.
}
\label{fig:system_model}
\end{figure}

Consider a downlink transmission from a base station (BS) to a single user
equipment (UE) located at $\boldsymbol{\psi}_{\rm u}\!\triangleq\!(x_{\rm u},
y_{\rm u}, 0)$ on the ground plane, as illustrated in
Fig.~\ref{fig:system_model}. The BS feeds a dielectric waveguide of total
length $L$ deployed horizontally at height $d$ along the $x$-axis, with
feed point at $x\!=\!x_f$ (taken as $x_f\!=\!0$ in the canonical single-feed
configuration and $x_f\!=\!L/2$ in the proposed SF-PASS
extension). A cluster of
$N_{\rm act}$ pinching antennas (PAs) is activated, with the $n$-th PA
($n\!=\!1,\dots,N_{\rm act}$) located at
\begin{equation}
\boldsymbol{\psi}_n^{\rm PA} \triangleq (x_{\rm c}\!+\!\delta_n,\, 0,\, d),\;
\delta_n \triangleq \left(n - \tfrac{N_{\rm act}+1}{2}\right)\lambda_{\rm g}(f_c)/2,
\label{eq:pa_positions}
\end{equation}
where $\lambda_{\rm g}(f_c)\!\triangleq\!c/(n_{\rm eff}(f_c) f_c)$ is the
guided wavelength, $n_{\rm eff}(f)$ is the effective (group) refractive
index of the dielectric waveguide (typically $n_{\rm eff}\!\approx\!1.4$
for THz porous-core
fibers~\cite{Ref_atakaramians2013thz,Ref_wang2025modeling}), and the
symmetric indexing makes $x_{\rm c}$ the cluster centroid. The
half-guided-wavelength spacing enables phase-coherent combining at
$f_c$~\cite{Ref_ding2025flexible,Ref_liu2026passtutorial}; as in the
canonical PASS model, each pinch is treated as a point-like dielectric
perturbation whose physical extent is sub-$\lambda_{\rm g}/2$, which at
sub-mm THz scale dictates sub-millimeter pinch fabrication---an active
THz dielectric-device design constraint~\cite{Ref_atakaramians2013thz}
but not a fundamental limitation. The cluster aperture
$(N_{\rm act}\!-\!1)\lambda_{\rm g}(f_c)/2$ is millimeter-scale at
THz---negligible relative to $L$, $1/\alpha_{\rm g}(f_c)$, and the
cluster-to-UE distance
$r_{\rm c}\!\triangleq\!\sqrt{(x_{\rm c}\!-\!x_{\rm u})^2\!+\!y_{\rm u}^2\!+\!d^2}$---so
all PAs share approximately the same large-scale path loss and
atmospheric attenuation, differing only in in-WG phases across
the OFDM band. To keep the entire cluster on the waveguide, the
cluster centroid is constrained to $x_{\rm c}\!\in\![\Delta_{\rm c},\,
L\!-\!\Delta_{\rm c}]$ with cluster half-aperture
$\Delta_{\rm c}\!\triangleq\!(N_{\rm act}\!-\!1)\lambda_{\rm g}(f_c)/4$,
i.e., a guard of $\Delta_{\rm c}$ at each waveguide end. The transmitted
waveform is OFDM with total bandwidth $B$ centered at $f_c$ and $M$
subcarriers indexed by $m\!\in\!\{1,\dots,M\}$ with frequencies \cite{Ref_xiao2025ofdmpass}
\begin{equation}
f_m \triangleq f_c + \left(m - \tfrac{M+1}{2}\right)\Delta f,\qquad
\Delta f \triangleq B/M.
\label{eq:subcarriers}
\end{equation}

\setcounter{equation}{6}
\begin{figure*}[!t]
\normalsize
\begin{equation}
y_m = \sqrt{\frac{G_{\rm t}\,G_{\rm r}\,P_{\rm t}(f_m)}{N_{\rm act}}}\,
\underbrace{\frac{c}{4\pi f_m r_{\rm c}}}_{\text{free-space spreading}}\,
\underbrace{e^{-\alpha_{\rm g}(f_m)\, x_{\rm c}}}_{\text{(i) in-WG loss}}\,
\underbrace{e^{-\frac{1}{2}\alpha_{\rm atm}(f_m)\, r_{\rm c}}}_{\text{(ii) atm.\ absorp.}}\,
\underbrace{\Bigg(\sum_{n=1}^{N_{\rm act}} e^{-j\beta(f_m)(x_{\rm c}+\delta_n)}\Bigg)}_{\text{(iv) beam-squint phasor}}\,
e^{-j\frac{2\pi f_m r_{\rm c}}{c}}\, s_m
\;+\; \underbrace{w_m}_{\text{(iii) thermal + mol.\ noise}}
\label{eq:received_signal_master}
\end{equation}
\hrulefill
\vspace*{4pt}
\end{figure*}
\setcounter{equation}{2}
\section{Wideband THz-PASS Propagation and Signal Modeling}\label{sec:channel_model}

This section develops a unified per-subcarrier channel and signal model
that integrates the four THz-unique physical effects: frequency-dependent
in-WG attenuation, atmospheric absorption, molecular re-radiation
noise, and beam squint across the PA cluster. We first build up the
physical components, and conclude with a single unified per-subcarrier
expression in which all four effects are marked
in-line.

The guided wave propagates along the waveguide with frequency-dependent
complex propagation constant
\begin{equation}
\gamma(f) \,\triangleq\, \alpha_{\rm g}(f) + j\,\beta(f),\quad
\beta(f) \,\triangleq\, \frac{2\pi\,n_{\rm eff}(f)\,f}{c},
\label{eq:gamma_wg}
\end{equation}
where $\alpha_{\rm g}(f)$ is the field-amplitude decay coefficient
($e^{-\alpha_{\rm g} x}$ field, $e^{-2\alpha_{\rm g} x}$ power), with
formulas using Np/m and numerical values reported in dB/m via
$\alpha[\mathrm{dB/m}]\!=\!8.686\,\alpha[\mathrm{Np/m}]$ (same convention
for atmospheric absorption $\alpha_{\rm atm}$). Realizable THz dielectric-waveguide
topologies~\cite{Ref_atakaramians2013thz} span the magnitude orders from $10^{-1}$ to $10^{2}$~dB/m,
contrasting sharply with the $\alpha_{\rm g}\!\approx\!0.08$~dB/m commonly
assumed at mmWave~\cite{Ref_xu2025inwgblockage}. Each pinch acts as a
passive radiator that couples a fraction of the guided power into free space
via evanescent leakage. Adopting the standard \emph{equal-power
model}~\cite{Ref_wang2025modeling,Ref_ding2025flexible,Ref_liu2026passtutorial},
the per-PA coupling coefficient is tuned to
$\kappa_n^2(f_m)\!\triangleq\!1/(N_{\rm act}\!-\!n\!+\!1)$, so that all
activated PAs in the cluster radiate equal power
\begin{equation}
P_n(f_m) = \frac{P_{\rm t}(f_m)}{N_{\rm act}}\,e^{-2\alpha_{\rm g}(f_m)(x_{\rm c}+\delta_n)} \approx \frac{P_{\rm t}(f_m)}{N_{\rm act}}\,e^{-2\alpha_{\rm g}(f_m)\,x_{\rm c}},
\label{eq:equal_power}
\end{equation}
where $P_{\rm t}(f_m)\!\triangleq\!P_{\rm t}/M$ is the per-subcarrier
input power under a sum-power constraint, and the small-cluster
approximation $\alpha_{\rm g}(f_m)|\delta_n|\!\ll\!1$ has been applied
in the second equality. The total radiated power
$\sum_{n}\!P_n(f_m)\!\approx\!P_{\rm t}(f_m)\,e^{-2\alpha_{\rm
g}(f_m) x_{\rm c}}$ is reduced from the input by the
$2\alpha_{\rm g} x_{\rm c}$~[Np] in-WG loss. The alternative proportional-power
model~\cite{Ref_wang2025modeling}, which uses a single shared $\kappa^2$
and yields a geometrically decaying per-PA power profile, gives
nearly-identical aggregate spectral efficiency in our operating
regime~\cite[Fig.~8]{Ref_wang2025modeling} but is less amenable to the
closed-form derivations of Sections~\ref{sec:beam_squint}
and~\ref{sec:optimal_placement}.

The composite per-PA channel coefficient $h_{n,m}(x_{\rm c})$ reads
\begin{equation}
h_{n,m}(x_{\rm c}) = \frac{c\,e^{-\frac{1}{2}\alpha_{\rm atm}(f_m) r_{\rm c}}}{4\pi f_m r_{\rm c}}\,
e^{-j[\beta(f_m)(x_{\rm c}+\delta_n) + 2\pi f_m r_{\rm c}/c]},
\label{eq:free_space}
\end{equation}
where the small-cluster approximation
$\|\boldsymbol{\psi}_n^{\rm PA}\!-\!\boldsymbol{\psi}_{\rm u}\|\!\approx\!r_{\rm c}$
has been applied to the magnitude terms. The specific atmospheric attenuation
$\alpha_{\rm atm}(f)$ is evaluated by line-by-line summation
of oxygen and water-vapor resonances under the standard atmosphere
$(T\!=\!15^{\circ}{\rm C}, P\!=\!1013.25~\mathrm{hPa},
\rho\!=\!7.5~\mathrm{g/m^3})$ following~\cite{Ref_itur_p676_13}.

Finally, the additive receiver noise comprises thermal noise and
molecular re-radiation noise~\cite{Ref_jornet2011channel}: atmospheric
molecules absorbing THz photons re-emit thermal radiation in turn,
contributing a distance- and frequency-dependent noise floor
\begin{equation}
N_{\rm total}(f_m, r_{\rm c}) \,=\, k_B T_{\rm sys}
\,+\, k_B T_0\Big(1 - e^{-K(f_m)\,r_{\rm c}}\Big),
\label{eq:noise}
\end{equation}
where $K(f)\!\triangleq\!2\alpha_{\rm atm}(f)$ is the power-domain molecular
absorption coefficient, $k_B$ is Boltzmann's constant, $T_{\rm sys}$ the
receiver system temperature, and $T_0\!=\!296.0$~K the reference
atmospheric temperature. The molecular term
$k_B T_0(1\!-\!e^{-K(f_m) r_{\rm c}})$ asymptotes to $k_B T_0$ as either
$K(f_m)$ or $r_{\rm c}$ grows, capturing the blackbody-like noise
contribution of an optically thick atmospheric path. This term, absent
from existing PASS analyses operating at mmWave, is small at lower THz
windows but becomes comparable to the thermal floor at $\geq\!1$~THz
where atmospheric absorption is high.

After coherent combining of the $N_{\rm act}$ PAs and substitution
of~\eqref{eq:equal_power}--\eqref{eq:free_space}, the received signal
at the UE on subcarrier $m$ admits the \emph{unified}
expression~\eqref{eq:received_signal_master} given at the top of this
page, in which all four THz-unique physical effects are marked in line. Here $s_m\!\in\!\mathbb{C}$ is the unit-energy data
symbol on subcarrier $m$, $G_{\rm t}$ and $G_{\rm r}$ are the BS and UE
antenna gains, and
$w_m\!\sim\!\mathcal{CN}(0, N_{\rm total}(f_m, r_{\rm c})\Delta f)$.
Dividing the signal-component power in~\eqref{eq:received_signal_master}
by the noise variance yields the per-subcarrier SNR at the UE,
\setcounter{equation}{7}
\begin{equation}
\gamma_m(x_{\rm c}) \,=\, \frac{P_{\rm t}(f_m)\,G_{\rm t}\,G_{\rm r}\,N_{\rm act}\,
|A_m(x_{\rm c})|^2\,\widetilde G_m(x_{\rm c})}{N_{\rm total}(f_m, r_{\rm c})\,\Delta f},
\label{eq:sinr_per_subc}
\end{equation}
where the placement-dependent effective channel gain
\begin{equation}
\widetilde G_m(x_{\rm c}) \,\triangleq\, \frac{G_{\rm FS}(f_m)\,
e^{-2\alpha_{\rm g}(f_m)\,x_{\rm c}\,-\,\alpha_{\rm atm}(f_m)\,r_{\rm c}}}{r_{\rm c}^2},
\label{eq:eff_gain}
\end{equation}
with free-space gain factor
$G_{\rm FS}(f)\!\triangleq\!\big(c/(4\pi f)\big)^2$, and the normalized
array factor
\begin{equation}
A_m(x_{\rm c}) \,\triangleq\, \frac{1}{N_{\rm act}}\sum_{n=1}^{N_{\rm act}}
e^{-j\,\Delta\phi_{n,m}},\quad
\Delta\phi_{n,m} \,\triangleq\, \phi_{n,m} - \phi_{n,c},
\label{eq:array_factor_def}
\end{equation}
captures the per-subcarrier residual beam-squint phase deviation
relative to the $f_c$-aligned configuration, with
$|A_m|^2\!\in\![0,1]$ and $|A_c|^2\!=\!1$ at $f\!=\!f_c$.
The next two sections derive closed-form expressions for (i) the
band-averaged coherence factor
$\bar\rho\!\triangleq\!(1/B)\!\int_{-B/2}^{+B/2}\!|A_m|^2\,d\Delta f_m$
(Section~\ref{sec:beam_squint}) and (ii) the cluster-center placement
$x_{\rm c}^{\star}$ that maximizes the worst-subcarrier $\gamma_m$
(Section~\ref{sec:optimal_placement}).

\section{Coherence Factor for Clustered Multi-PA}\label{sec:beam_squint}

By the half-guided-wavelength spacing
$\delta_n\!=\!(n\!-\!(N_{\rm act}\!+\!1)/2)\,\lambda_{\rm g}(f_c)/2$, the in-WG phase difference at the carrier $f_c$ is an integer multiple of $\pi$ and can be removed by a
common phase compensation. At an off-center subcarrier $f_m\!=\!f_c\!+\!\Delta
f_m$, neglecting the higher-order $n_{\rm eff}(f)$-dispersion term that is
small within a single THz window, the residual squint phase per PA is
(up to a common offset)
\begin{equation}
\delta_{\phi_{n,m}} = \frac{\pi(n\!-\!(N_{\rm act}\!+\!1)/2)\,\Delta f_m}{f_c}.
\label{eq:residual_phase}
\end{equation}
The normalized array factor $A_m(x_{\rm c})$ in
\eqref{eq:array_factor_def} therefore admits the closed
form
\begin{equation}
|A_m|^2 = \frac{1}{N_{\rm act}^2}\,\left|\frac{\sin\!\big(N_{\rm act}\,\pi\,\Delta f_m / (2 f_c)\big)}{\sin\!\big(\pi\,\Delta f_m / (2 f_c)\big)}\right|^{2}.
\label{eq:dirichlet}
\end{equation}
Notably, $|A_m|^2$ is independent of the cluster center $x_{\rm c}$. Under the
small-cluster assumption, beam squint thus depends only on the array
configuration $(N_{\rm act}, B/f_c)$, not on the cluster position; this
decoupling is the key structural simplification that makes the
placement problem in Section~\ref{sec:optimal_placement} tractable.

\begin{theorem}[Closed-Form Coherence Factor]\label{thm:coherence}
The band-averaged array gain (coherence factor) for clustered multi-PA
THz-PASS is
\begin{equation}
\bar\rho \,\triangleq\, \frac{1}{B}\!\int_{-B/2}^{+B/2}\!\!|A_m|^2\,d\Delta f_m \;\approx\; 1 - \frac{\pi^2\,(N_{\rm act}^2\!-\!1)\,B^2}{144\,f_c^2},
\label{eq:rhobar_closedform}
\end{equation}
asymptotically tight for $N_{\rm act}\,B/f_c\!\ll\!1$.
\end{theorem}
\begin{IEEEproof}
Taylor-expanding the Dirichlet kernel at $u\!=\!\pi\Delta f_m/(2f_c)$
gives $|A_m|^2\!\approx\!1\!-\!(N_{\rm act}^2\!-\!1)u^2/3$; integrating
over $\Delta f_m\!\in\![-B/2, B/2]$ yields~\eqref{eq:rhobar_closedform}.
\end{IEEEproof}

\section{Optimal Cluster-Center Placement and Inversion Threshold}\label{sec:optimal_placement}

Substituting Theorem~\ref{thm:coherence} into~\eqref{eq:sinr_per_subc} and
noting that $\bar\rho\,N_{\rm act}$ is a positive multiplicative constant
independent of $x_{\rm c}$, the cluster-center optimization problem reduces to
\begin{equation}
x_{\rm c}^{\star} \,=\, \arg\!\max_{x_{\rm c}\in[\Delta_{\rm c},\,L-\Delta_{\rm c}]}\;\min_{m\in[M]}\;\widetilde\gamma_m(x_{\rm c}),
\label{eq:p1}
\end{equation}
where
\begin{equation}
\widetilde\gamma_m(x_{\rm c}) = \frac{G_{\rm FS}(f_m)\,e^{-2\alpha_{\rm g}(f_m)\,x_{\rm c} - \alpha_{\rm atm}(f_m)\,r_{\rm c}(x_{\rm c})}}{r_{\rm c}(x_{\rm c})^2\,N_{\rm total}(f_m, r_{\rm c}(x_{\rm c}))}.
\label{eq:p1_obj}
\end{equation}
The max-min objective robustly captures the worst-subcarrier SINR, which
at THz resides either at a band edge or at a dominant atmospheric
resonance peak; $\widetilde\gamma_m(x_{\rm c})$ is unimodal in $x_{\rm c}$
(log-convex $1/r_{\rm c}^2$ $\times$ monotone $e^{-2\alpha_{\rm g} x_{\rm c}}$).

\begin{theorem}[SINR-Equalization Condition]\label{thm:xc_star}
Assume the worst-SINR subcarriers at the optimum are the band edges
$\{f_{\rm low}, f_{\rm high}\}\!=\!\{f_c\!-\!B/2,\,f_c\!+\!B/2\}$. Then the
optimum $x_{\rm c}^{\star}$ of~\eqref{eq:p1} satisfies the SINR-equalization
condition
$\widetilde\gamma_{\rm low}(x_{\rm c}^{\star})
\!=\!\widetilde\gamma_{\rm high}(x_{\rm c}^{\star})$, equivalent to
\begin{equation}
2\,\Delta\alpha_{\rm g}\,x_{\rm c}^{\star} + \Delta\alpha_{\rm atm}\,r_{\rm c}(x_{\rm c}^{\star})
\,=\, \Delta\!\ln G - \Delta\!\ln N\!\big(r_{\rm c}(x_{\rm c}^{\star})\big),
\label{eq:eq8}
\end{equation}
where the band-edge differentials are
$\Delta\alpha_{\rm g}\!\triangleq\!\alpha_{\rm g}(f_{\rm high})\!-\!\alpha_{\rm g}(f_{\rm low})$,
$\Delta\alpha_{\rm atm}\!\triangleq\!\alpha_{\rm atm}(f_{\rm high})\!-\!\alpha_{\rm atm}(f_{\rm low})$,
$\Delta\!\ln G\!\triangleq\!\ln[G_{\rm FS}(f_{\rm low})/G_{\rm FS}(f_{\rm high})]$,
and
$\Delta\!\ln N(r)\!\triangleq\!\ln[N_{\rm total}(f_{\rm high}, r)/N_{\rm total}(f_{\rm low}, r)]$.

Under the user-projection regime $x_{\rm c}^{\star}\!\approx\!x_{\rm u}$,
$r_{\rm c}\!\approx\!d_{\min}\!\triangleq\!\sqrt{y_{\rm u}^2\!+\!d^2}$,
the explicit closed form ($\Delta\alpha_{\rm g}\!>\!0$) is
\begin{equation}
x_{\rm c}^{\star} \,\approx\, \frac{\Delta\!\ln G - \Delta\!\ln N(d_{\min}) - \Delta\alpha_{\rm atm}\,d_{\min}}{2\,\Delta\alpha_{\rm g}}.
\label{eq:eq9}
\end{equation}
Otherwise,~\eqref{eq:eq8} is solved by 1D bisection over
$[\Delta_{\rm c},\,L\!-\!\Delta_{\rm c}]$.
\end{theorem}
\begin{IEEEproof}
Unimodality of $\widetilde\gamma_m$ ensures the max-min is attained where
the two binding band-edge subcarriers achieve equal SINR; taking
$\ln\widetilde\gamma_{\rm low}\!=\!\ln\widetilde\gamma_{\rm high}$ and
substituting $r_{\rm c}\!\approx\!d_{\min}$ gives~\eqref{eq:eq9}.
\end{IEEEproof}

\begin{corollary}[Placement-Inversion Threshold]\label{cor:inversion}
Define the regime-change threshold $\alpha_{\rm g}^{\star\star}$ as the value
of $\Delta\alpha_{\rm g}$ at which~\eqref{eq:eq8} balances with
$x_{\rm c}\!=\!x_{\rm u}$:
\begin{equation}
\alpha_{\rm g}^{\star\star} = \frac{\Delta\!\ln G - \Delta\!\ln N(d_{\min}) - \Delta\alpha_{\rm atm}\,d_{\min}}{2\,x_{\rm u}}.
\label{eq:alpha_star_star}
\end{equation}
For $\Delta\alpha_{\rm g}\!<\!\alpha_{\rm g}^{\star\star}$, the optimum lies in
\textit{Regime~A} ($x_{\rm c}^{\star}\!\approx\!x_{\rm u}$, ``placed above user''),
recovering the existing PASS placement
rule~\cite{Ref_ding2025flexible,Ref_xu2025inwgblockage}. For
$\Delta\alpha_{\rm g}\!>\!\alpha_{\rm g}^{\star\star}$, the optimum shifts toward
the waveguide feed point ($x_{\rm c}^{\star}\!\to\!0$, \textit{Regime~B}). In a
clean transparency window with negligible atmospheric and noise differentials,
\eqref{eq:alpha_star_star} simplifies to
\begin{equation}
\alpha_{\rm g}^{\star\star} \,\approx\, \frac{\ln(1+B/f_c)}{x_{\rm u}} \,\approx\, \frac{B}{f_c\,x_{\rm u}}\;\;\text{for}\;\;B/f_c\!\ll\!1.
\label{eq:alpha_star_star_clean}
\end{equation}
\end{corollary}

\section{Switched-Feed PASS}\label{sec:sfpass}

The placement inversion of Corollary~\ref{cor:inversion} forces
$x_{\rm c}^{\star}\!\to\!0$ at high $\alpha_{\rm g}$ because every meter
of cluster offset from a fixed single end-feed extracts $2\alpha_{\rm g}$
Np of in-WG loss. We address this with \emph{SF-PASS}:
a centrally located RF switch of insertion loss $L_{\rm sw}$
($1$--$3$~dB at THz) routes the BS RF chain into one of $K\!\geq\!2$
contiguous waveguide segments. For a 1D corridor of length $L$, the
canonical $K\!=\!2$ instance feeds at $x_f\!=\!L/2$ with segments
$\mathcal{S}_1\!=\![0,L/2]$, $\mathcal{S}_2\!=\![L/2,L]$; for a 2D
side-$L$ room, $K\!=\!4$ feeds at room center, routing into four arms
$\{\mathcal{S}_1,\ldots,\mathcal{S}_4\}$ of length $L/2$. Let
$k^{\star}$ index the segment containing the user. The cluster is then
constrained to $\boldsymbol{\psi}_{\rm c}\!\in\!\mathcal{S}_{k^{\star}}$,
the feed-to-cluster in-WG distance becomes
$\|\boldsymbol{\psi}_{\rm c}\!-\!\boldsymbol{\psi}_f\|$ (in place of
$x_{\rm c}$), and the cluster-to-user distance generalizes to
$r_{\rm c}\!\triangleq\!\|\boldsymbol{\psi}_{\rm c}\!-\!\boldsymbol{\psi}_{\rm u}\|$
(reducing to the 1D form for the corridor). Substituting
into~\eqref{eq:eff_gain} and absorbing the band-edge-invariant
$L_{\rm sw}$ as a constant prefactor gives
\begin{equation}
\widetilde G_m^{\rm SF}(\boldsymbol{\psi}_{\rm c}) \,=\, L_{\rm sw}^{-1}\,
\frac{G_{\rm FS}(f_m)\,e^{-2\alpha_{\rm g}(f_m)\|\boldsymbol{\psi}_{\rm c}-\boldsymbol{\psi}_f\|-\alpha_{\rm atm}(f_m) r_{\rm c}}}{r_{\rm c}^2}.
\label{eq:eff_gain_sfpass}
\end{equation}

\begin{theorem}[SF-PASS Placement]\label{thm:sfpass}
The SF-PASS optimum
$\boldsymbol{\psi}_{\rm c}^{\star,{\rm SF}}\!\in\!\mathcal{S}_{k^{\star}}$
satisfies the SINR-equalization condition of Theorem~\ref{thm:xc_star}
with $x_{\rm c}\!\to\!\|\boldsymbol{\psi}_{\rm c}\!-\!\boldsymbol{\psi}_f\|$,
giving
\begin{equation}
\|\boldsymbol{\psi}_{\rm c}^{\star,{\rm SF}}\!-\!\boldsymbol{\psi}_f\| \,\approx\,
\frac{\Delta\!\ln G - \Delta\!\ln N(d_{\min}) - \Delta\alpha_{\rm atm}\,d_{\min}}{2\,\Delta\alpha_{\rm g}}.
\label{eq:xc_sfpass}
\end{equation}
\end{theorem}
\begin{IEEEproof}
$L_{\rm sw}$ is band-edge invariant and cancels from SINR equalization;
Theorem~\ref{thm:xc_star}'s proof carries over under the stated
substitution.
\end{IEEEproof}

\begin{corollary}[SF-PASS Inversion and Payoff]\label{cor:sfpass}
SF-PASS halves the worst-case feed-to-cluster distance from $L$ to $L/2$
(both 1D $K\!=\!2$ and 2D $K\!=\!4$), lowering the inversion threshold
to
\begin{equation}
\alpha_{\rm g}^{\star\star,\,{\rm SF}} \,\approx\,
\frac{2\,\ln(1\!+\!B/f_c)}{L}.
\label{eq:alpha_star_sfpass}
\end{equation}
SF-PASS outperforms single-feed PASS when in-WG savings exceed
$L_{\rm sw}$:
\begin{equation}
\alpha_{\rm g} \,>\, \alpha_{\rm g}^{(\rm sw)} \,\triangleq\,
\frac{L_{\rm sw}}{x_{\rm u}\!-\!L/2}
\quad\text{(}L_{\rm sw}\text{ in dB, distance in m).}
\label{eq:payoff_threshold}
\end{equation}
For $L\!=\!50$~m, $x_{\rm u}\!=\!40$~m, $L_{\rm sw}\!=\!2$~dB:
$\alpha_{\rm g}^{(\rm sw)}\!\approx\!0.133$~dB/m, exceeded by every
realizable THz waveguide in~\cite{Ref_atakaramians2013thz} (e.g.,
porous-COC at $0.22$~THz, $1.3$~dB/m, $\sim\!10\times$ above).
\end{corollary}

\section{Numerical Results}\label{sec:numerical}

We adopt a long-corridor scenario ($L\!=\!50$~m, $d\!=\!3$~m), UE at
$x_{\rm u}\!=\!40$~m near the corridor end, PA cluster
$N_{\rm act}\!=\!16$ under the equal-power model, OFDM with $M\!=\!64$
subcarriers and $B\!=\!10$~GHz at $f_c\!=\!300$~GHz, antenna gains
$G_{\rm t}\!=\!20$~dBi, $G_{\rm r}\!=\!0$~dBi. The primary dielectric
waveguide is porous-core polyethylene (PE) with the verified
$\alpha_{\rm g}\!=\!0.43$~dB/m at $0.3$~THz~\cite[Table~5]{Ref_atakaramians2013thz}, representing the
best-realized PASS-compatible THz waveguide; for the multi-band Theorem-2
analysis we also include porous-core COC ($1.3$~dB/m at $0.22$~THz) and
microstructured Teflon bandgap ($2.2$~dB/m at $1$~THz). Atmospheric
attenuation is evaluated line-by-line per ITU-R~P.676 yielding
$\alpha_{\rm atm}\!=\!4$~dB/km at $300$~GHz. Three benchmark architectures
are compared: an UPA at corridor center with
$M\!\in\!\{64, 256\}$ half-wavelength elements (single RF chain via
analog beamforming); and the proposed SF-PASS
with $K\!=\!2$ segments, $x_f\!=\!L/2$, and switch insertion loss
$L_{\rm sw}\!=\!2$~dB. All schemes operate under the same sum-power
constraint $P_{\rm t}$ and assume perfect CSI.

\emph{Theorem validation.} Theorem~\ref{thm:coherence}'s closed-form
gives $\bar\rho\!\approx\!0.98$ at the primary scenario ($f_c\!=\!300$~GHz,
$B\!=\!10$~GHz, $N_{\rm act}\!=\!16$)---squint penalty $<\!0.1$~dB, benign.
Under Theorem~\ref{thm:xc_star}, at 300~GHz porous-PE the optimum
$x_{\rm c}^{\star}\!\approx\!39.4$~m (near-user; Regime~A), and the
$\alpha_{\rm g} x_{\rm u}\!\approx\!17$~dB in-WG penalty drops single-feed
PASS to $2.69$~bps/Hz (see Fig.~\ref{fig:se_vs_pt}). At higher-loss
operating points (e.g., 1~THz microstr.~Teflon, 2.2~dB/m) the optimum
shifts toward $x_{\rm c}^{\star}\!\to\!0$ as predicted by
Corollary~\ref{cor:inversion}; SF-PASS recovers near-user proximity in
both regimes.

Fig.~\ref{fig:se_vs_pt} reports SE vs.~$P_{\rm t}$ at the corridor end
($x_{\rm u}\!=\!40$~m, 300~GHz porous-PE, $\alpha_{\rm g}\!=\!0.43$~dB/m;
practical THz operation today lies in $[0,15]$~dBm given current
InP/SiGe power amplifier  technology~\cite{Ref_jiang2024thz6g}).
At $P_{\rm t}\!=\!10$~dBm, single-feed PASS achieves $2.69$~bps/Hz
(Regime~A: $x_{\rm c}^{\star}\!=\!39.4$~m, indistinguishable from
fixed-PASS at $x_{\rm c}\!=\!x_{\rm u}\!=\!40$~m, so the two curves
coincide; they diverge only in the higher-loss Regime~B of
Fig.~\ref{fig:sfpass_payoff}), degraded by
$\alpha_{\rm g}\,x_{\rm u}\!\approx\!17$~dB in-WG loss.
\emph{SF-PASS $K\!=\!2$ doubles this to $5.36$~bps/Hz} by reducing the
feed-to-cluster distance from $x_{\rm u}\!=\!40$~m (single-feed at $x\!=\!0$)
to $\|\boldsymbol{\psi}_{\rm c}^{\star,{\rm SF}}\!-\!\boldsymbol{\psi}_f\|
\!\approx\!15$~m (centrally-fed switch at $x_f\!=\!L/2$), giving in-WG loss
$\approx\!6.5$~dB plus $L_{\rm sw}\!=\!2$~dB switch loss, and \emph{essentially ties the
UPA $M\!=\!64$} at $5.43$~bps/Hz---within
$0.07$~bps/Hz at $1/4$ the active-antenna count and using passive
dielectric pinches rather than active power amplifiers. The higher-cost
antenna-parity UPA $M\!=\!256$ delivers $7.40$~bps/Hz at the price of
$256$ active antennas and $256$ phase shifters.

\begin{figure}[t]
\centering
\includegraphics[width=0.65\columnwidth]{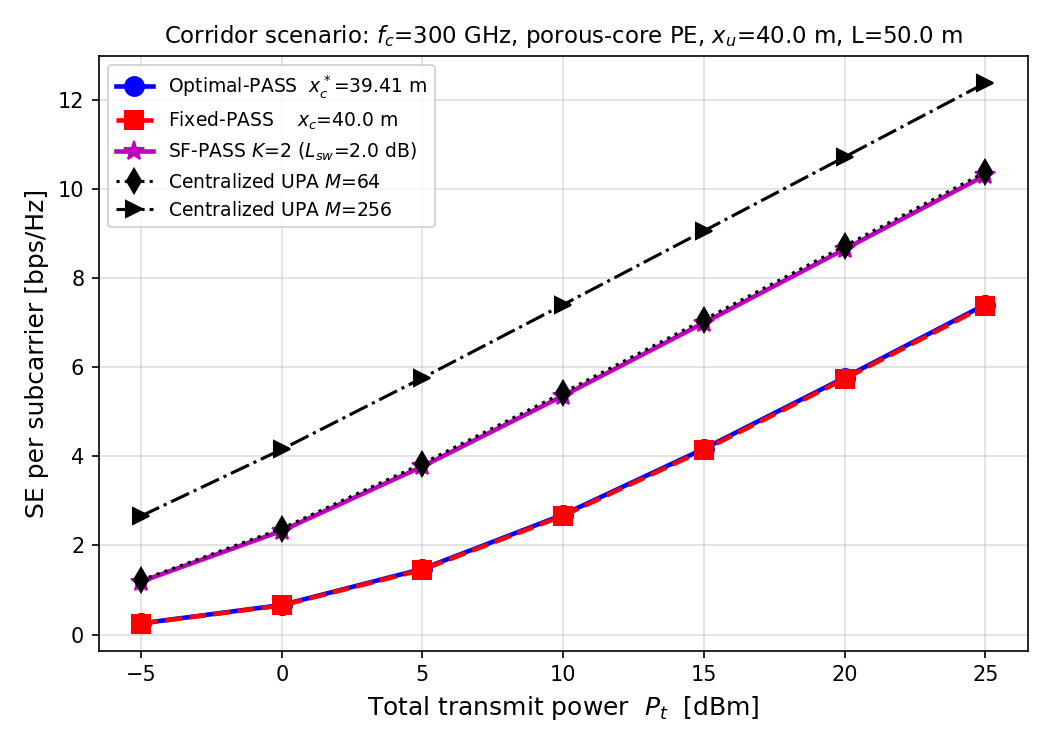}
\caption{Performance comparison: SE vs.~$P_{\rm t}$, where SF-PASS doubles
single-feed PASS and essentially ties the  UPA $M\!=\!64$.}
\label{fig:se_vs_pt}
\end{figure}

Fig.~\ref{fig:sfpass_payoff} sweeps $\alpha_{\rm g}$ over $[10^{-3},
10^{3}]$~dB/m at 300~GHz and confirms the payoff crossover at
$\alpha_{\rm g}^{(\rm sw)}\!\approx\!0.133$~dB/m: every PASS-compatible
waveguide reported in~\cite{Ref_atakaramians2013thz}
(300~GHz porous-PE $3.2\!\times$ above;
220~GHz porous-COC $9.8\!\times$ above; 1~THz microstr.~Teflon
$16.5\!\times$ above) operates in the SF-PASS-favored regime. The SF-PASS
SE gain over single-feed peaks near $\alpha_{\rm g}\!\sim\!0.3$~dB/m and
remains at $1.5$--$2.7$~bps/Hz across the realistic $[0.4, 50]$~dB/m
range.

\begin{figure}[t]
\centering
\includegraphics[width=0.65\columnwidth]{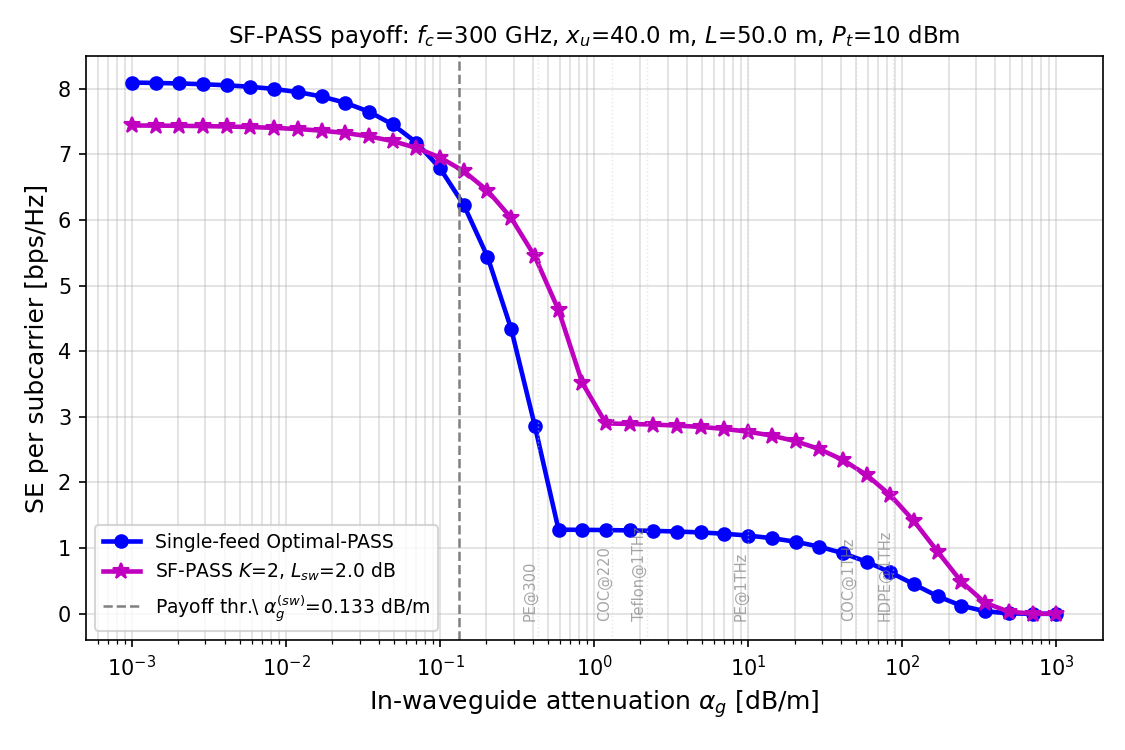}
\caption{SF-PASS payoff: SE vs.\ in-waveguide attenuation $\alpha_{\rm g}$.
Dashed vertical line marks the closed-form payoff threshold
$\alpha_{\rm g}^{(\rm sw)}\!\approx\!0.133$~dB/m.}
\label{fig:sfpass_payoff}
\end{figure}

SF-PASS's value lies in its hardware-cost class (Table~\ref{tab:cost}).
At THz, per-element cost is dominated by the active front-end---power
amplifier, high-speed DAC/ADC, mixer, LO chain---constrained by limited
InP/SiGe technology~\cite{Ref_jiang2024thz6g}. \emph{The UPA baseline
assumes the standard THz analog-beamforming architecture with one power
amplifier per antenna element placed after each phase shifter: per-amplifier output power is bounded
by InP/SiGe component capabilities at THz (5--10~dBm at 300~GHz), and the alternative
single-power-amplifier-plus-passive-splitter configuration incurs prohibitive
$10\log_{10}(M)\!\approx\!18$~dB splitter loss for $M\!=\!64$.}
Consequently, the UPA $M\!=\!64$ requires $64$ active THz power
amplifiers and $64$ phase shifters; UPA $M\!=\!256$ requires $256$ of
each. SF-PASS reduces the entire front-end to one RF chain $+$ one THz
power amplifier $+$ a passive dielectric waveguide bearing
$N_{\rm act}\!=\!16$ pinch elements. \emph{SF-PASS matches the SE of
UPA $M\!=\!64$} at 300~GHz porous-PE using $64\times$ fewer active power
amplifiers ($1$ vs.\ $64$), validating the single-RF-chain
single-power-amplifier THz-PASS architecture as a viable low-cost
alternative.

\begin{table}[t]
\centering
\caption{Hardware-Cost Comparison.}
\label{tab:cost}
\renewcommand{\arraystretch}{1.1}
\setlength{\tabcolsep}{4pt}
\footnotesize
\begin{tabular}{lcccc}
\hline
\textbf{Architecture} & \textbf{RF} & \textbf{Power Amplifier} & \textbf{Phase Shifter} & \textbf{Antennas} \\
\hline
UPA $M\!=\!64$          & 1 & 64  & 64  & 64 active \\
UPA $M\!=\!256$    & 1 & 256 & 256 & 256 active \\
\text{SF-PASS}       & \text{1} & \textbf{1} & \text{0} & \text{16 passive} \\
\hline
\end{tabular}
\end{table}

\section{Conclusion}\label{sec:conclusion}

We presented a unified wideband THz-PASS framework and Switched-Feed
PASS architecture with five closed-form results
(Theorems~\ref{thm:coherence}--\ref{thm:sfpass},
Corollaries~\ref{cor:inversion}--\ref{cor:sfpass}). At the best
PASS-compatible waveguide reported in~\cite{Ref_atakaramians2013thz}
(300~GHz porous-PE, $\alpha_{\rm g}\!=\!0.43$~dB/m), SF-PASS with two segments \emph{doubles}
single-feed PASS  (with spectral efficiency of $2.69\!\to\!5.36$~bps/Hz at the corridor end) by
halving the effective in-waveguide distance, and \emph{essentially ties}
the UPA with $64$ active elements ($5.43$~bps/Hz) using a single
RF chain feeding a single THz power amplifier and $N_{\rm act}\!=\!16$ passive
dielectric pinch elements---$64\times$ fewer active power amplifiers than the
 UPA. This positions SF-PASS as the architectural enabler of
single-RF-chain THz-PASS at corridor scale. Extensions include
numerical evaluation of more scenarios like the 2D-room geometry, multi-user
wideband PASS-THz modeling and optimization, mobility-adaptive placement, and integrated PASS-THz sensing and communications.

\bibliographystyle{IEEEtran}
\bibliography{IEEEabrv,Ref_THz_PASS}

\begin{thebibliography}{10}
\providecommand{\url}[1]{#1}
\csname url@samestyle\endcsname
\providecommand{\newblock}{\relax}
\providecommand{\bibinfo}[2]{#2}
\providecommand{\BIBentrySTDinterwordspacing}{\spaceskip=0pt\relax}
\providecommand{\BIBentryALTinterwordstretchfactor}{4}
\providecommand{\BIBentryALTinterwordspacing}{\spaceskip=\fontdimen2\font plus
\BIBentryALTinterwordstretchfactor\fontdimen3\font minus \fontdimen4\font\relax}
\providecommand{\BIBforeignlanguage}[2]{{%
\expandafter\ifx\csname l@#1\endcsname\relax
\typeout{** WARNING: IEEEtran.bst: No hyphenation pattern has been}%
\typeout{** loaded for the language `#1'. Using the pattern for}%
\typeout{** the default language instead.}%
\else
\language=\csname l@#1\endcsname
\fi
#2}}
\providecommand{\BIBdecl}{\relax}
\BIBdecl

\bibitem{Ref_jiang2024thz6g}
W.~Jiang \emph{et~al.}, ``Terahertz communications and sensing for {6G} and beyond: A comprehensive review,'' \emph{{IEEE} Commun. Surveys Tuts.}, vol.~26, no.~4, pp. 2326--2381, Oct. 2024.

\bibitem{Ref_jornet2011channel}
J.~M. Jornet and I.~F. Akyildiz, ``Channel modeling and capacity analysis for electromagnetic wireless nanonetworks in the terahertz band,'' \emph{{IEEE} Trans. Wireless Commun.}, vol.~10, no.~10, pp. 3211--3221, Oct. 2011.

\bibitem{Ref_fukuda2022pinching}
A.~Fukuda, H.~Yamamoto, H.~Okazaki, Y.~Suzuki, and K.~Kawai, ``Pinching antenna --- using a dielectric waveguide as an antenna,'' \emph{{NTT DOCOMO} Tech. J.}, vol.~23, no.~3, pp. 5--12, Jan. 2022.

\bibitem{Ref_ding2025flexible}
Z.~Ding, R.~Schober, and H.~V. Poor, ``Flexible-antenna systems: A pinching-antenna perspective,'' \emph{{IEEE} Trans. Commun.}, vol.~73, no.~10, pp. 9236--9253, Oct. 2025.

\bibitem{Ref_liu2026passtutorial}
Y.~Liu \emph{et~al.}, ``Pinching-antenna systems ({PASS}): A tutorial,'' \emph{{IEEE} Trans. Commun.}, vol.~74, pp. 4881--4918, 2026.

\bibitem{Ref_wang2025modeling}
Z.~Wang, C.~Ouyang, X.~Mu, Y.~Liu, and Z.~Ding, ``Modeling and beamforming optimization for pinching-antenna systems,'' \emph{{IEEE} Trans. Commun.}, vol.~73, no.~12, pp. 13\,904--13\,919, Dec. 2025.

\bibitem{Ref_xu2025inwgblockage}
Y.~Xu, Z.~Ding, O.~A. Dobre, and T.-H. Chang, ``Pinching-antenna system design with {LoS} blockage: Does in-waveguide attenuation matter?'' \emph{arXiv:2508.07131}, Aug. 2025.

\bibitem{Ref_atakaramians2013thz}
S.~Atakaramians, S.~{Afshar V.}, T.~M. Monro, and D.~Abbott, ``Terahertz dielectric waveguides,'' \emph{Adv. Opt. Photon.}, vol.~5, no.~2, pp. 169--215, Jul. 2013.

\bibitem{Ref_itur_p676_13}
{International Telecommunication Union}, \emph{Attenuation by Atmospheric Gases and Related Effects}, {ITU-R} Recommendation {P.676-13}, Aug. 2022.

\bibitem{Ref_tarboush2021teramimo}
S.~Tarboush \emph{et~al.}, ``{TeraMIMO}: A channel simulator for wideband ultra-massive {MIMO} terahertz communications,'' \emph{{IEEE} Trans. Veh. Technol.}, vol.~70, no.~12, pp. 12\,325--12\,341, Dec. 2021.

\bibitem{Ref_xiao2025ofdmpass}
J.~Xiao, J.~Wang, M.~Zeng, Y.~Liu, and G.~K. Karagiannidis, ``Frequency-selective modeling and analysis for {OFDM}-integrated wideband pinching-antenna systems,'' \emph{{IEEE} Wireless Commun. Lett.}, vol.~14, no.~11, pp. 3500 -- 3504, Nov. 2025.

\end{thebibliography}

\end{document}